\documentclass{article}

\usepackage{PRIMEarxiv}

\usepackage[utf8]{inputenc} 
\usepackage[T1]{fontenc}    

\usepackage{amsmath}
\usepackage{amssymb}
\usepackage{amsfonts}
\usepackage{amsthm}

\usepackage{graphicx}       
\usepackage{tikz}
\usepackage{pgfplots}
\usepackage{subcaption}

\usepackage{booktabs}       
\usepackage{multirow}
\usepackage{threeparttable}

\usepackage[ruled]{algorithm2e}

\usepackage{microtype}      
\usepackage{enumitem}
\usepackage{fancyhdr}       
\usepackage{caption}
\usepackage{boxedminipage}

\usepackage{authblk}
\usepackage{comment}
\usepackage{lipsum}
\usepackage{nicefrac}
\usepackage{pifont}
\usepackage{thm-restate}
\usepackage{xr}

\usepackage{natbib}

\usepackage{hyperref}
\usepackage{url}
\hypersetup{
    colorlinks=true,
    linkcolor=blue,
    citecolor=blue,
    urlcolor=blue
}
\usepackage{cleveref}

\usepackage{xcolor}

\newtheorem{definition}{Definition}%
\newtheorem{theorem}{Theorem}%
\newtheorem{example}{Example}

\allowdisplaybreaks

\title{\bf MenuNet: A Strategy-Proof Mechanism for Matching Markets}

\author[1,2]{Zhaohong Sun}
\author[1]{Makoto Yokoo}

\affil[1]{Kyushu University, Japan}
\affil[2]{CyberAgent, Japan}

\date{\vspace{-10mm}}

\begin{document}

\maketitle

\begin{abstract}
Strategy-proofness is a fundamental desideratum in mechanism design, ensuring truthful reporting and robust participation. Stability is another central requirement in matching markets, widely adopted in applications such as school choice and labor market clearing. In practice, however, these markets are invariably governed by complex distributional constraints, ranging from diversity quotas and regional balance to global capacity slacks, under which stable matchings often fail to exist. This raises a fundamental question: how to distribute unavoidable instability across agents while preserving strategy-proofness? To address this, we propose \texttt{MenuNet}, a strategy-proof mechanism design framework based on a neural representation of menus. Rather than directly constructing assignments, \texttt{MenuNet} learns to generate personalized probabilistic menus, from which assignments are realized via a structured sequential choice rule that guarantees strategy-proofness by construction. By decomposing stability into fairness (no envy) and non-wastefulness, our approach models these properties as vector-valued quantities and optimizes their distribution through differentiable objectives, providing a principled trade-off between competing axioms. Empirically, \texttt{MenuNet} navigates this trade-off effectively: it consistently outperforms Random Serial Dictatorship (RSD) in terms of envy and Deferred Acceptance (DA) in terms of waste, while maintaining scalability and computational efficiency. These results suggest that learning-based menu mechanisms provide a flexible and scalable paradigm for mechanism design in highly constrained, real-world environments.
\end{abstract}

\section{Introduction}
Strategy-proofness (SP) is a central desideratum in matching markets, ensuring that agents have no incentive to misreport their preferences. By aligning incentives with truthful reporting, SP simplifies strategic behavior and enhances the robustness of the mechanism. Its importance is well illustrated by major market design reforms, such as the National Resident Matching Program (NRMP) in the U.S. \citep{Roth84a} and the Boston school choice system \citep{APRS05a}, where manipulable mechanisms were replaced by strategy-proof alternatives. These reforms are largely based on the Deferred Acceptance (DA) mechanism \citep{GaSh62a}, which is strategy-proof for the proposing side and produces stable outcomes.
Stability is another fundamental requirement, ensuring that no pair of agents can profitably deviate to form a blocking match. This property is widely adopted in practice, including medical residency matching \citep{RoPe99a,Roth08a}, school choice systems in New York City \citep{APR05a} and Chile \citep{CEE+22a}, centralized college admissions in India \citep{BCC+19b}, and daycare matching in Japan \citep{KaKo24a, STM+23a, SYT+24a}.

In many real-world applications, however, matching markets are subject to distributional and institutional constraints that extend beyond the classical framework of rigid capacities. For instance, lower quotas may require institutions to enroll a minimum number of students to remain viable \citep{BFIM10a}, while regional quotas impose aggregate capacity limits across groups of institutions \citep{KaKo15a}. Diversity and affirmative action policies introduce type-based priorities to enroll students from diverse backgrounds \citep{SoYe22a, AyBo21a}, and complementarities arise in settings such as residency matching for couples \citep{McMa10a, KPR13a}. More generally, multi-dimensional feasibility constraints encountered in refugee resettlement further complicate the solution space \citep{DKT23a, ACGS18a}.

In the presence of such constraints, stable matchings are no longer guaranteed to exist. The existing literature has largely addressed this challenge by decomposing stability into fairness (absence of justified envy) and non-wastefulness, subsequently relaxing these properties to restore feasibility under specific constraints. However, these approaches typically focus on aggregate objectives, offering limited control over the distributional consequences of the resulting instability. In practice, this can lead to systemic imbalance, where a small subset of agents bear the brunt of the violations, resulting in extreme individual dissatisfaction.
Since some degree of instability is unavoidable in constrained environments, the key design question is no longer how to eliminate instability, but how to distribute it. Our goal is to design strategy-proof mechanisms that not only balance fairness and non-wastefulness, but also equitably partition unavoidable instability across agents to preclude disproportionate individual loss.

To address this challenge, we introduce \texttt{MenuNet}, a neural mechanism design framework built upon personalized probabilistic menus. Rather than directly optimizing discrete assignments, \texttt{MenuNet} generates an admission probability menu for each agent, from which final allocations are induced via a structured choice rule that ensures strategy-proofness by construction (Theorem~\ref{theo:MenuNet:SP}). 
The framework is modular and versatile, modeling both envy and waste as vector-valued quantities. This formulation allows \texttt{MenuNet} to optimize differentiable objectives with fine-grained control, facilitating a principled trade-off between fairness and efficiency while preventing the concentration of instability from disproportionately affecting a small subset of agents. Moreover, the architecture naturally accommodates a broad spectrum of distributional constraints, and is computationally efficient and scales to large markets.

To demonstrate the effectiveness of our approach, we focus on a representative constraint, global capacity slack, where capacity constraints are treated as flexible and allow for controlled and bounded violations. Such settings commonly arise in practice: universities may over-admit to hedge against yield uncertainty; research labs may temporarily exceed supervision limits at a marginal cost; and course allocation systems often expand high-demand sections to better satisfy preferences. These requirements introduce interdependent constraints across institutions, leading to environments in which capacity is sufficient in aggregate but remains locally scarce.

Through extensive experiments across diverse market scales, we demonstrate that \texttt{MenuNet} achieves a superior trade-off between fairness and non-wastefulness compared to classical baselines such as Random Serial Dictatorship (RSD) and Deferred Acceptance (DA). Specifically, \texttt{MenuNet} consistently attains lower envy than RSD and lower waste than DA, while crucially distributing unavoidable dissatisfaction more equitably across the agent population. Furthermore, the framework exhibits robust scalability, maintaining computational efficiency as it scales to large-scale markets.

\section{Related Work}
Recent advances in neural mechanism design have demonstrated the potential of deep learning for optimal auction design \citep{FNP18a,RJBW21a,DTYF+22a,ISFB22a,DFN+24a}. In contrast, comparatively little work has applied neural approaches to matching problems. A notable exception is \citep{RFL+21a}, which studies neural mechanisms for classical two-sided matching. Our work differs from this paper in several important respects. First, while it focuses on classical matching environments with rigid capacities, our framework is designed for settings with flexible quotas and more general distributional constraints commonly arising in institutional applications. Second, rather than directly learning assignments through regret-based objectives, \texttt{MenuNet} adopts a menu-based representation that guarantees strategy-proofness by construction, providing a structural alternative to empirical regularization. Third, \texttt{MenuNet} is designed for scalability, achieving efficient training and inference even in large markets. 

A separate line of work studies matching under uncertainty through bandit-based frameworks, where agent preferences must be learned online \citep{DaKa05a,LMJ20a,LRMJ21a,Basu25a}. These works address preference elicitation and exploration-exploitation trade-offs, whereas our focus is on mechanism design under known preferences but complex feasibility constraints.

\section{Model}

In this section, we formalize the school choice model with global capacity flexibility, a specific form of distributional constraints. While our proposed \texttt{MenuNet} framework applies to a broad class of matching markets, we focus on this setting for illustrative purposes.

An instance is defined by the tuple $\mathcal{I} = (S, C, \mathbf{q}, k, V, U)$, where $S = \{s_1, \dots, s_n\}$ and $C = \{c_0, c_1, \dots, c_m\}$ denote the sets of students and schools, respectively. Each school $c \in C \setminus \{c_0\}$ is associated with a \emph{soft capacity target} $q_c \in \mathbb{Z}_{\geq 0}$, representing its preferred maximum enrollment. We include an outside option $c_0 \in C$, representing the option of being unassigned, and assume $c_0$ has sufficient capacity to accommodate all students (i.e., $q_{c_0} \ge n$). To capture flexibility in the system, we introduce a \emph{global slack} parameter $k \in \mathbb{Z}_{\geq 0}$, which bounds the total allowable over-enrollment across all schools.

To enable gradient-based optimization, we adopt a cardinal representation of preferences and priorities. Let $V = [v_{s,c}] \in \mathbb{R}^{n \times (m+1)}$ and $U = [u_{s,c}] \in \mathbb{R}^{n \times (m+1)}$ denote the student utility matrix and the school priority matrix, respectively, where $v_{s,c}$ represents the utility of student $s$ for school $c$, and $u_{s,c}$ represents the priority assigned to $s$ by school $c$.
All scores for acceptable assignments are normalized to lie in $[0,1]$, with higher values indicating stronger preferences or higher priorities. We allow $v_{s,c} < 0$ to indicate that school $c$ is unacceptable to student $s$. The outside option is normalized to $v_{s,c_0} = 0$, ensuring that all unacceptable schools are strictly dominated.\footnote{We assume all students are acceptable to all schools from the schools' perspective, as is standard in the school choice literature; this assumption does not affect the design of our algorithm.}


A matching $\mu$ assigns each student to exactly one element of $C$, i.e., $\mu(s) \in C$ for all $s \in S$. For each school $c \in C$, let $\mu(c) = \{s \in S \mid \mu(s) = c\}$ denote the set of students assigned to $c$. A matching is \emph{feasible} if the aggregate capacity violation satisfies:
\begin{equation}
\sum_{c \in C \setminus \{c_0\}} \max\{|\mu(c)| - q_c, 0\} \leq k.
\end{equation}

%
Stability is a central concept in matching markets, defined as the combination of \emph{individual rationality} and the absence of \emph{blocking pairs} \citep{Roth85a}. Individual rationality requires that no student is assigned to an unacceptable school, while a pair $(s,c)$ blocks a matching if both sides prefer to be matched with each other over their current assignments.
In our setting, blocking pairs can be decomposed into two distinct types: either a student can displace a lower-priority student at a school, or can be assigned to a school with available capacity without violating feasibility. These two cases correspond, respectively, to violations of \emph{fairness} and \emph{non-wastefulness}, standard notions in the literature on matching with distributional constraints \citep{KaKo24a}.
\begin{definition}[Fairness]
\label{def:fair}
A matching $\mu$ is \emph{fair} if there does not exist a pair $(s,c)$ such that (i) $s$ strictly prefers $c$ to $\mu(s)$, and (ii) there exists a student $s' \in \mu(c)$ with lower priority than $s$ at $c$.
\end{definition}
\begin{definition}[Non-wastefulness]
\label{def:NW}
A matching $\mu$ is \emph{non-wasteful} if there does not exist a pair $(s,c)$ such that 
(i) $s$ strictly prefers $c$ to $\mu(s)$, and 
(ii) the matching $\mu'$, obtained by reassigning $s$ to $c$ and leaving all other students' assignments unchanged, is feasible.
\end{definition}
The inherent tension between fairness and non\mbox{-}wastefulness under global capacity constraints can be illustrated by the following example.
\begin{example}
\label{example:non-existence}
Consider an instance with two students $S=\{s_1,s_2\}$ and two schools $C=\{c_1,c_2\}$. The system is subject to a global quota $k=1$, so that at most one student can be assigned in any feasible matching. Preferences and priorities are perfectly misaligned: student $s_1$ prefers $c_2$ over $c_1$, and $s_2$ prefers $c_1$ over $c_2$, while school $c_1$ ranks $s_1$ above $s_2$, and $c_2$ ranks $s_2$ above $s_1$. We claim that no deterministic matching satisfies both fairness and non-wastefulness.
If we assign $s_1$ to $c_1$, the matching is wasteful because $s_1$ prefers $c_2$ and reassigning $s_1$ to $c_2$ remains feasible ($k=1$). If we instead assign $s_1$ to $c_2$, fairness is violated because $s_2$ has higher priority at $c_2$ and prefers $c_2$ to being unmatched.
By symmetry, assigning $s_2$ to $c_2$ is wasteful, while assigning $s_2$ to $c_1$ violates fairness. Therefore, every non-empty deterministic matching either violates fairness or is wasteful.
\end{example}
The impossibility result in Example~\ref{example:non-existence} reveals a fundamental limitation of deterministic mechanisms under global constraints. This motivates the introduction of randomization as a relaxation of deterministic assignments. To illustrate, let $\mu^{(1)}$ be the matching with $\mu^{(1)}(s_1)=c_2$, and let $\mu^{(2)}$ be the matching with $\mu^{(2)}(s_2)=c_1$. The lottery $P = \tfrac{1}{2}\mu^{(1)} + \tfrac{1}{2}\mu^{(2)}$ assigns each student to their top choice with probability $1/2$. While such randomization does not eliminate instability ex post, it redistributes instability more evenly across agents, thereby mitigating systematic disadvantage and yielding a more balanced notion of ex-ante fairness.

%
A random assignment is represented by a marginal probability matrix $P \in [0,1]^{|S| \times |C|}$, where each entry $P_{s,c}$ denotes the probability that student $s$ is assigned to school $c$. Let $\mathcal{F}$ denote the set of feasible random assignments if it satisfies the following two conditions: each student is assigned to exactly one option (including the outside option $c_0$), and the global capacity constraint is satisfied in expectation up to slack $k$. Formally,
\begin{equation}
\mathcal{F} = \left\{ P \in [0,1]^{|S|\times|C|} \;\middle|\; \sum_{c \in C} P_{s,c} = 1, \forall s \in S; \quad \sum_{c \in C \setminus \{c_0\}} \max\left(0, \sum_{s \in S} P_{s,c} - q_c\right) \le k \right\}.
\end{equation}
%
%
%

We next introduce differentiable notions of ex-ante envy and ex-ante waste, which respectively extend fairness and non-wastefulness to probabilistic assignments.
For any students $s,s' \in S$ and school $c \in C$, we define the \emph{ex-ante envy intensity} of $s$ toward $s'$ at school $c$ under assignment matrix $P$ as:
\begin{equation}
\label{eq:envy_intensity}
    e_{s,s',c}(P) = P_{s',c} \cdot \max\left\{ v_{s,c} - \bar{V}_s(P), \, 0 \right\} \cdot \mathbb{I}(u_{s,c} > u_{s',c}),
\end{equation}
where $\bar{V}_s(P) = \sum_{d \in C} P_{s,d} v_{s,d}$ denotes the expected utility of student $s$ under the probabilistic assignment $P$, and $\mathbb{I}(\cdot)$ is the indicator function. This intensity captures three critical factors: (i) the probability mass $P_{s',c}$ allocated to student $s'$, (ii) the potential marginal gain in expected utility if $s$ were assigned to $c$ instead, and (iii) the priority condition $\mathbb{I}(u_{s,c} > u_{s',c})$ which ensures that envy is only \emph{justified} when $s$ has a strictly higher priority than $s'$ at school $c$.

Aggregating over all students $s' \in S$ and all schools $c \in C$,  we define the total ex-ante envy experienced by student $s$ as the normalized intensity across the population:
\begin{equation}
\label{eq:envy_instensity}
    E_s(P) = \frac{1}{|S|}\sum_{s' \in S}\sum_{c \in C} e_{s,s',c}(P).
\end{equation}
Let $E(P) = (E_s(P))_{s \in S} \in \mathbb{R}_+^{|S|}$ denote the resulting envy vector. This construction captures how ex-ante envy is distributed across students under random assignments. Unlike discrete notions of justified envy, the continuity of $E(P)$ allows the mechanism to minimize not only the aggregate magnitude of justified envy but also to penalize its distribution, ensuring that no single agent suffers disproportionately, especially when global constraints render perfect fairness unattainable.

We extend non\mbox{-}wastefulness to random assignments through a slack-modulated intensity formulation. Let $\ell_c(P) = \sum_{s \in S} P_{s,c}$ denote the expected load at school $c$. Define the total overflow as $\Omega(P) = \sum_{c \in C} \max\{\ell_c(P) - q_c, 0\}$. 
Slack arises from two complementary sources. The \emph{local slack} at school $c$ is given by $s_c^{\mathrm{loc}}(P) = \max\{q_c - \ell_c(P), 0\}$, capturing unused target capacity. The \emph{global slack} is given by $s^{\mathrm{glob}}(P) = \max\{k - \Omega(P), 0\}$, capturing the remaining allowable over-enrollment under the global slack. 
The ex-ante waste intensity for student $s$ at school $c$ is then defined as
\begin{equation}
\label{eq:waste_intensity}
    w_{s,c}(P) = \max\left\{ v_{s,c} - \bar{V}_s(P), \, 0 \right\}  \cdot
    \min\!\left\{1,\, s_c^{\mathrm{loc}}(P) + s^{\mathrm{glob}}(P) \right\}
\end{equation}
where $\bar{V}_s(P)$ is the expected utility as defined previously. The waste intensity captures the potential utility gain for student $s$ from being assigned to school $c$, scaled by the extent to which such an assignment is feasible under the available slack. By averaging over all schools, we define the total ex-ante waste for student $s$ as 
\begin{equation}
\label{envy_vector}
W_s(P) = \frac{1}{|C|} \sum_{c \in C} w_{s,c}(P)
\end{equation}
Let $W(P) = (W_s(P))_{s \in S} \in \mathbb{R}_+^{|S|}$ denote the resulting \emph{waste vector}. Together, $E(P)$ and $W(P)$ provide a unified, differentiable representation of ex-ante instability, characterizing how fairness and non\mbox{-}wastefulness violations are distributed across agents.

Given an instance $\mathcal{I}$, let $\mathcal{M}: \mathcal{V} \mapsto \Delta(\mathcal{X})$ denote a mechanism that maps reported preference profiles to a probability distribution over feasible matchings $\mathcal{X}$. A mechanism $\mathcal{M}$ is \emph{strategyproof} (SP) if, for every student $s \in S$, any true preference profile $V$, and any possible misreport $V'_s$, it holds that:
\begin{equation}
\label{eq:SP_definition}
\mathbb{E}_{\mu \sim \mathcal{M}(V_s, V_{-s})} \left[ v_{s,\mu(s)} \right]
\;\geq\;
\mathbb{E}_{\mu \sim \mathcal{M}(V'_s, V_{-s})} \left[ v_{s,\mu(s)} \right],
\end{equation}
where $V_{-s}$ denotes the reported preferences of all students other than $s$, and the expectation is taken over the internal randomness of the mechanism. Under this definition, truthful reporting is a dominant strategy for all students, ensuring that no agent can gain a higher expected utility by misrepresenting their preferences.

\section{MenuNet Mechanism}

We begin by introducing a straightforward deterministic menu mechanism, which serves as the basis of our neural approach. Instead of directly producing assignments, the mechanism decouples the allocation process into two stages: \emph{menu design} and \emph{individual choice}. For each student, the mechanism  specifies a menu of available schools, from which the student selects their most preferred option. To ensure strategy-proofness, the menu must be constructed independently of each student's own reported preferences, although it may depend on the reports of other students. However, the mechanism does not guarantee feasibility, as multiple students may select the same popular schools, potentially resulting in assignments that violate the capacity constraints.



\subsection{MenuNet: A Neural Menu Mechanism}

We next introduce \texttt{MenuNet}, a neural mechanism that generates personalized probabilistic menus and optimizes system-level objectives via differentiable loss functions. This framework preserves strategy-proofness by construction while enabling the mechanism to internalize complex feasibility and stability constraints through end-to-end training. 
At a high level, \texttt{MenuNet}consists of two stages: (i) \emph{menu generation}, where each student is assigned a personalized probability distribution over schools, and (ii) \emph{assignment}, where the final allocation is realized based on the student’s preferences and the generated menu. 

\textbf{Menu Generation:} For each student $s \in S$, we construct a personalized \emph{probabilistic menu} using a shared neural network $\mathcal{F}_\theta$ parameterized by $\theta$. The input to the network captures the global market state, including the priority matrix $U$, the capacity vector $q$, and the student preferences. To satisfy strategy-proofness, we adopt a leave-one-out construction for the preference input: for each student $s$, the network only observes the preferences of all other agents, denoted by $V_{-s}$. The personalized menu for student $s$ is thus given by 
\[
p_s = \mathcal{F}_\theta(V_{-s}, U, q) \in [0,1]^m,
\]
where each component $p_{s,c}$ represents the probability that school $c$ is available to student $s$. We fix the availability of the outside option to $p_{s,c_0} = 1$, ensuring that every student has at least one guaranteed feasible selection.

\textbf{Assignment:} Given the probabilistic menu $p_s$, the final assignment probability is realized through a sequential choice rule that reflects the student's preference ranking. For each student $s$, let the schools be indexed such that $c_1 \succ_s c_2 \succ_s \cdots \succ_s c_m$ according to their reported utilities. The marginal probability $P_{s,c_j}$ of student $s$ being assigned to their $j$-th preferred school $c_j$ is defined as the joint probability that $c_j$ is available and all more preferred schools $\{c_k\}_{k < j}$ are unavailable:
\begin{equation}
P_{s,c_j} = p_{s,c_j} \prod_{k < j} (1 - p_{s,c_k}).
\end{equation}
By this construction, the total assigned probability satisfies $\sum_{c \in C} P_{s,c} \le 1$, with any residual mass implicitly allocated to the outside option. To ensure numerical stability during the backward pass and avoid vanishing gradients in long preference chains, we compute the assignment probabilities in the log-domain:
\begin{equation}
\log P_{s,c_j} = \log p_{s,c_j} + \sum_{k < j} \log(1 - p_{s,c_k}).
\end{equation}
The resulting assignment matrix $P$ is a continuous and differentiable function of the menu parameters $p_s$, enabling the use of gradient-based optimization to minimize the envy and waste losses defined previously.

\begin{theorem}
\label{theo:MenuNet:SP}
The MenuNet mechanism is strategyproof for students.
\end{theorem}

\begin{proof}
To prove that the \texttt{MenuNet} mechanism $\mathcal{M}$ is strategyproof, we must show that for any student $s \in S$, reporting true preferences $V_s$ maximizes their expected utility, satisfying the condition:
\begin{equation}
\mathbb{E}_{\mu \sim \mathcal{M}(V_s, V_{-s})} \left[ v_{s,\mu(s)} \right] \geq \mathbb{E}_{\mu \sim \mathcal{M}(V'_s, V_{-s})} \left[ v_{s,\mu(s)} \right], \forall V'_s.
\end{equation}

The mechanism $\mathcal{M}$ decouples the allocation into two stages. First, a probabilistic menu $p_s$ is generated via the neural network $\mathcal{F}_\theta(V_{-s}, U, q)$. By the \emph{leave-one-out} architecture, $p_s$ is functionally independent of student $s$'s own report. Thus, for any misreport $V'_s$, the generated menu remains constant: $p_s(V_s, V_{-s}) = p_s(V'_s, V_{-s}) = p_s$. 

Second, given a fixed $p_s$, student $s$ is assigned to schools based on their reported ranking. Let the reported ranking be $c_1 \succ_s c_2 \succ_s \cdots \succ_s c_m$. The expected utility of student $s$ is given by:
\begin{equation}
U_s(V_s \mid p_s) = \sum_{j=1}^m \left( p_{s,c_j} \prod_{k < j} (1 - p_{s,c_k}) \right) v_{s,c_j}.
\end{equation}
This assignment rule can be viewed as a sequential application process where the student applies to schools in the reported order, and each school $c$ admits the student independently with probability $p_{s,c}$ if reached. We show that $U_s$ is maximized if and only if the student reports the true descending order of utilities.

Consider any two schools $c$ and $c'$ that are adjacent in the reported ranking, appearing at positions $j$ and $j+1$. Let $\alpha = \prod_{k < j} (1 - p_{s,c_k})$ be the probability that the student is not admitted by any school ranked higher than $j$.
If $c$ is ranked before $c'$, the contribution of these two positions to the expected utility is:
\begin{equation}
\Delta = \alpha \left[ p_{s,c} v_{s,c} + (1 - p_{s,c}) p_{s,c'} v_{s,c'} \right].
\end{equation}
If the student swaps the order of $c$ and $c'$, the new contribution becomes:
\begin{equation}
\Delta' = \alpha \left[ p_{s,c'} v_{s,c'} + (1 - p_{s,c'}) p_{s,c} v_{s,c} \right].
\end{equation}
The difference between the two cases is:
\begin{equation}
\Delta - \Delta' = \alpha \, p_{s,c} p_{s,c'} \left( v_{s,c} - v_{s,c'} \right).
\end{equation}
Since $\alpha, p_{s,c}, p_{s,c'} \in [0,1]$, the sign of the difference is determined solely by $(v_{s,c} - v_{s,c'})$. Specifically, if $v_{s,c} \geq v_{s,c'}$, placing $c$ before $c'$ weakly increases the expected utility. 

Any misreport $V'_s$ that induces an ordering different from the true utility ranking can be transformed back to the truthful ranking through a finite sequence of such adjacent swaps. Since each swap toward the true order non-decreasingly improves the expected utility, the truthful report $V_s$ must be an optimal strategy. Because the menu $p_s$ is invariant to $V_s$, student $s$ has no incentive to misreport, completing the proof.
\end{proof}

\subsection{Optimization Objectives and Training}
In our framework, market instances are characterized by the tuple $(V, U)$, where student preferences and school priorities are sampled from a distribution representative of the target market. The network $\mathcal{F}_\theta$ is trained offline to minimize a composite objective over a large set of sampled instances. A key advantage of this approach is that after training, the mechanism requires no instance-specific optimization; for any newly encountered market, the assignment matrix $P$ is obtained via a single forward pass of the network. This ensures that the mechanism is computationally efficient and scalable to real-time allocation problems.

We train \texttt{MenuNet} by optimizing a loss function $\mathcal{L}(P)$ that reflects the core desiderata of market design. For a given assignment $P$, the total loss is defined as:
\begin{equation}
\mathcal{L}(P) = \lambda_{\mathrm{w}} \mathcal{L}_{\mathrm{welf}}(P) + \lambda_{\mathrm{c}} \mathcal{L}_{\mathrm{capa}}(P) + \lambda_{\mathrm{s}} \mathcal{L}_{\mathrm{stab}}(P),
\end{equation}
where the coefficients $\lambda$ control the trade-offs between competing objectives. Specifically, $\mathcal{L}_{\mathrm{welf}}$ represents the negative social welfare, $\mathcal{L}_{\mathrm{capa}}$ penalizes violations of the global capacity constraint $k$, and $\mathcal{L}_{\mathrm{stab}}$ captures the magnitude and distribution of instability through the ex-ante envy and waste vectors defined previously.

\textbf{Welfare Loss.}
To encourage assignments that align with student preferences and promote social efficiency, we incorporate an aggregate welfare metric into the learning objective. Minimizing $\mathcal{L}_{\mathrm{welf}}$ pushes the mechanism toward allocations where students receive higher-ranked schools in expectation. We define the welfare loss as the negative average expected utility across all students:
\begin{equation}
\mathcal{L}_{\mathrm{welf}}(P) = - \frac{1}{|S|} \sum_{s \in S} \bar{V}_s(P) = - \frac{1}{|S|} \sum_{s \in S} \sum_{c \in C} P_{s,c} \, v_{s,c}.
\end{equation}

\textbf{Feasibility Loss: The $K$-Barrier Penalty.}
To integrate the global capacity constraint into a gradient-based framework, we employ a smooth penalty based on the total expected overflow. Let $\ell_c(P) = \sum_{s \in S} P_{s,c}$ denote the expected load at school $c$, and define the total overflow as 
\[
\Omega(P) = \sum_{c \in C} \max\{ \ell_c(P) - q_c,\, 0 \}.
\]
Given the permissible slack $k$, we define the \emph{$K$-barrier penalty} as a piecewise function:
\begin{equation}
\mathcal{L}_{\mathrm{capa}}(P) = 
\begin{cases} 
\alpha \, \Omega(P), & \text{if } \Omega(P) \le k, \\
\alpha k + \beta (\Omega(P) - k)^2 + \gamma (\Omega(P) - k), & \text{if } \Omega(P) > k,
\end{cases}
\end{equation}
where $\alpha, \beta, \gamma > 0$ are hyperparameters controlling the penalty strength. When the overflow remains within the allowable slack, the penalty is linear, encouraging the mechanism to efficiently utilize the flexible capacity. Once the overflow exceeds the threshold $k$, the penalty transitions to a \emph{super-linear regime}, combining quadratic and linear terms to sharply penalize further violations. In our implementation, we set $\gamma = 10$ to ensure a sufficient gradient signal at the transition boundary.

\textbf{Stability Loss.}
Under global feasibility constraints, it is generally impossible to eliminate all blocking pairs, as illustrated by the non-existence result in Example~\ref{example:non-existence}. We therefore shift our focus from total elimination to regulating the \emph{distributional profile} of instability. We characterize instability via two metrics: the envy vector $E(P) \in \mathbb{R}_+^{|S|}$ and the waste vector $W(P) \in \mathbb{R}_+^{|S|}$, which capture ex-ante violations of fairness and efficiency, respectively.
For each component, we adopt a structured aggregation combining three principles: \emph{egalitarian} (reducing total magnitude), \emph{Rawlsian} (protecting the worst-off agent), and \emph{balanced} (discouraging concentration). The resulting losses $\mathcal{L}_{\mathrm{envy}}$ and $\mathcal{L}_{\mathrm{waste}}$ are defined as weighted combinations of these terms, enabling control over both the magnitude and distribution of instability.

\textbf{Egalitarian Principle: } This objective aims to reduce the aggregate magnitude of dissatisfaction across the population. We employ the mean squared error (MSE) of the vectors to penalize larger individual violations more heavily than smaller ones:
\begin{equation}
\label{eq:loss:ega}
\mathcal{L}_{\mathrm{egal}}^{E}(P) = \frac{1}{|S|} \sum_{s \in S} \big(E_s(P)\big)^2, \quad \mathcal{L}_{\mathrm{egal}}^{W}(P) = \frac{1}{|S|} \sum_{s \in S} \big(W_s(P)\big)^2.
\end{equation}

\textbf{Rawlsian Principle: } This principle prioritizes the worst-off agents by minimizing the maximum dissatisfaction experienced by any single student. Since the max operator is non-differentiable and unsuitable for gradient-based optimization, we approximate the $L_\infty$ norm (the maximum operator) using a smooth \emph{log-sum-exp} operator:
\begin{equation}
\mathcal{L}_{\mathrm{rawl}}^{E}(P) = \tau_E \log \sum_{s \in S} \exp\left(\frac{E_s(P)}{\tau_E}\right), \quad \mathcal{L}_{\mathrm{rawl}}^{W}(P) = \tau_W \log \sum_{s \in S} \exp\left(\frac{W_s(P)}{\tau_W}\right),
\end{equation}
where $\tau > 0$ is a temperature parameter. As $\tau \to 0$, this objective converges to the classical minimax criterion, effectively shielding individuals from extreme instability.

\textbf{Balanced Principle: } To prevent certain agents from bearing a disproportionate share of instability, we define the balanced objectives as the variance of the dissatisfaction vectors:
\begin{equation}
\label{eq:loss:variance}
\mathcal{L}_{\mathrm{bal}}^{E}(P) = \operatorname{Var}(E(P)), \quad \mathcal{L}_{\mathrm{bal}}^{W}(P) = \operatorname{Var}(W(P)).
\end{equation}
By penalizing the variance, we encourage the mechanism to distribute unavoidable envy and waste more uniformly across the student body.

The overall stability loss $\mathcal{L}_{\mathrm{stab}}(P)$ is the weighted sum of these components, allowing for flexible control over the trade-offs between magnitude and fairness:
\begin{align}
    \mathcal{L}_{\mathrm{envy}}(P) &= \lambda_1^E \mathcal{L}_{\mathrm{egal}}^E + \lambda_2^E \mathcal{L}_{\mathrm{rawl}}^E + \lambda_3^E \mathcal{L}_{\mathrm{bal}}^E, \\
    \mathcal{L}_{\mathrm{waste}}(P) &= \lambda_1^W \mathcal{L}_{\mathrm{egal}}^W + \lambda_2^W \mathcal{L}_{\mathrm{rawl}}^W + \lambda_3^W \mathcal{L}_{\mathrm{bal}}^W.
\end{align}
The final objective is given by $\mathcal{L}_{\mathrm{stab}}(P) = \mathcal{L}_{\mathrm{envy}}(P) + \mathcal{L}_{\mathrm{waste}}(P)$, providing a comprehensive, differentiable measure of ex-ante instability for end-to-end training.

\section{Experiments}
We evaluate the performance of \texttt{MenuNet} through extensive simulations across diverse market scales. To benchmark its efficacy, we compare it against two foundational strategy-proof mechanisms: Random Serial Dictatorship (RSD), which prioritizes non-wastefulness, and Deferred Acceptance (DA), which prioritizes fairness under global constraints. The results indicate that \texttt{MenuNet} achieves a consistent trade-off between ex-ante envy and waste. In addition, it yields allocations in which instability is more evenly distributed across agents, avoiding the concentration of violations observed in the baseline mechanisms. From a computational perspective, the proposed approach is efficient and scalable. Training remains tractable on a CPU, requiring approximately 10 minutes for markets with 1{,}000 students. 

\subsection{Experimental Setup: Data Generation}
We evaluate the mechanism across a range of market scales, with the student population varying from $100$ to $2{,}000$ in increments of $100$, while fixing the number of schools at $10$. For each configuration, we generate $1{,}000$ training, $200$ validation, and $100$ test instances.

Student preferences are generated using the Mallows model \citep{LuCr11a}, which characterizes preference profiles as noisy perturbations centered around a reference ranking $\succ^{\mathrm{ref}}_s$. The dispersion parameter $\phi_{\mathrm{student}} \in (0,1]$ dictates the degree of heterogeneity; smaller values signify stronger correlation among student preferences, while $\phi_{\mathrm{student}} = 1$ corresponds to uniform randomness.  We refer to Appendix~\ref{appendix:mallows} for a formal description of the model and the sampling procedure.
To reflect realistic constraints on student choices, we implement truncated preferences where each student identifies only their top-$k$ schools as acceptable. Preference utilities $V$ are assigned as a decreasing function of rank for acceptable schools, supplemented by small random perturbations to facilitate tie-breaking. The outside option $c_0$ is consistently assigned zero utility.

School priorities are sampled using an analogous Mallows structure, where a reference ranking $\succ^{\mathrm{ref}}_c$ over students is perturbed by a dispersion parameter $\phi_{\mathrm{school}}$. Priority scores are derived from the resulting ranks, with identical priority assigned to students for the outside option $c_0$ given its unlimited capacity. 
The total capacity is set proportional to the number of students, with ratio $0.6$ in all experiments. Individual school capacities are sampled from a Gaussian distribution with mean equal to the average capacity per school and a specified standard deviation, and then discretely adjusted to ensure that the total capacity exactly matches the target level. 
We allow limited violations of capacity constraints via a global slack parameter $K$, which bounds the total overflow. In all experiments, we set $K = 0.05 \cdot |S|$, allowing up to $5\%$ of the number of students as aggregate over-enrollment.

\subsection{Implementation Details}
The neural mechanism $\mathcal{F}_\theta$ is parameterized as a three-layer multilayer perceptron (MLP) with $256$ hidden units per layer. To enhance training stability in large-scale matching environments, each hidden layer incorporates Layer Normalization and ReLU activation. The model is trained using the Adam optimizer with a learning rate of $1 \times 10^{-4}$, a batch size of $16$, and a total of $30$ epochs. To prevent gradient explosion and ensure numerical stability, we apply gradient clipping with a maximum norm of $1.0$. During training, the sequential survival-based construction uses a temperature parameter $\tau_{\mathrm{prio}} = 100.0$ for differentiable priority comparisons, while the final evaluation is performed under a fixed scale to ensure exactness.

All experiments are implemented using the PyTorch framework.The computational cost of the proposed approach scales approximately linearly with the market size, and this trend is consistent across both CPU and GPU implementations. For instance, for a market with $1,000$ students, the training procedure completes in approximately $10$ minutes on an Apple M4 Max CPU, and is further accelerated to approximately $1$ minute on an NVIDIA RTX 5090 GPU. The feasibility of CPU execution is a significant practical advantage, as it demonstrates that the mechanism is computationally tractable without the need for specialized hardware. This positioning suggests that GPU acceleration is an optional speedup rather than a fundamental requirement for scalability, making the approach accessible for a wide range of institutional matching applications. Detailed scaling curves for both CPU and GPU performance are provided in the appendix.

\subsection{Experimental Results}
We benchmark the performance of \texttt{MenuNet} against two foundational strategy-proof mechanisms: Random Serial Dictatorship (RSD) and Deferred Acceptance (DA). To ensure a fair comparison, both baselines are augmented to accommodate the global slack constraints as detailed in Appendix~\ref{appendix:baseline}.

We evaluate envy and waste across multiple aggregation criteria, including the population mean, egalitarian (MSE), Rawlsian (maximum), and variance-based measures. Across all metrics, \texttt{MenuNet} consistently achieves a superior trade-off between fairness and efficiency. Crucially, the mechanism distributes instability more evenly across the student population, effectively mitigating the concentration of violations observed in DA and RSD. This shift toward a more balanced profile of dissatisfaction is quantitatively depicted in Figure~\ref{fig:metrics} of the Appendix.

Regarding feasibility, \texttt{MenuNet} respects the global flexibility budget $K$ with high precision. As shown in Figure~\ref{fig:overflow} of the Appendix, the total expected overflow remains closely aligned with the prescribed limit, with only negligible deviations across a small subset of instances. Notably, we employ a consistent set of hyperparameters across all market sizes. While further tuning of the penalty weights for each specific market scale could potentially eliminate these minor slack violations entirely, our results demonstrate that a unified parameter configuration already provides robust and reliable performance.

In terms of social welfare, \texttt{MenuNet} exhibits a slight performance gap compared to RSD and DA. This outcome is anticipated, as the mechanism explicitly internalizes the costs of stability and global feasibility within its objective. These findings underscore the inherent tension between welfare maximization and distributional stability, reinforcing the value of a multi-objective approach that prioritizes system-wide fairness.

Overall, our experiments validate \texttt{MenuNet} as a principled framework for balancing competing desiderata in market design. It achieves substantial improvements in distributional stability while maintaining computational tractability and preserving the fundamental property of strategy-proofness.

\section{Conclusion}
We study mechanism design under complex distributional constraints, where stable matchings may fail to exist. We propose \texttt{MenuNet}, a neural framework that shifts the objective from feasibility to the principled distribution of unavoidable instability. By ensuring strategy-proofness by construction and optimizing vector-valued fairness objectives, \texttt{MenuNet} enables fine-grained control over the distribution of dissatisfaction. Experiments under global capacity slack show that it effectively balances fairness and non-wastefulness, outperforming classical benchmarks while remaining computationally efficient at scale. These results highlight the potential of learning-based menu mechanisms as a practical paradigm for constrained matching.

Several directions remain for future work. First, beyond standard baselines, it would be valuable to compare against specialized strategy-proof mechanisms tailored to structured constraints, such as those based on M$\natural$-convexity~\citep{KTY18a}. Second, designing more expressive and symmetry-aware neural architectures may further improve performance and generalization. Finally, extending the framework to richer classes of multi-dimensional distributional constraints is an important step toward broader applicability.

\newpage

\bibliographystyle{unsrtnat}
\bibliography{reference}  


\newpage
\appendix
\begin{figure}[tb]
    \centering

    \begin{subfigure}{0.9\textwidth}
        \centering
        \includegraphics[width=0.8\linewidth]{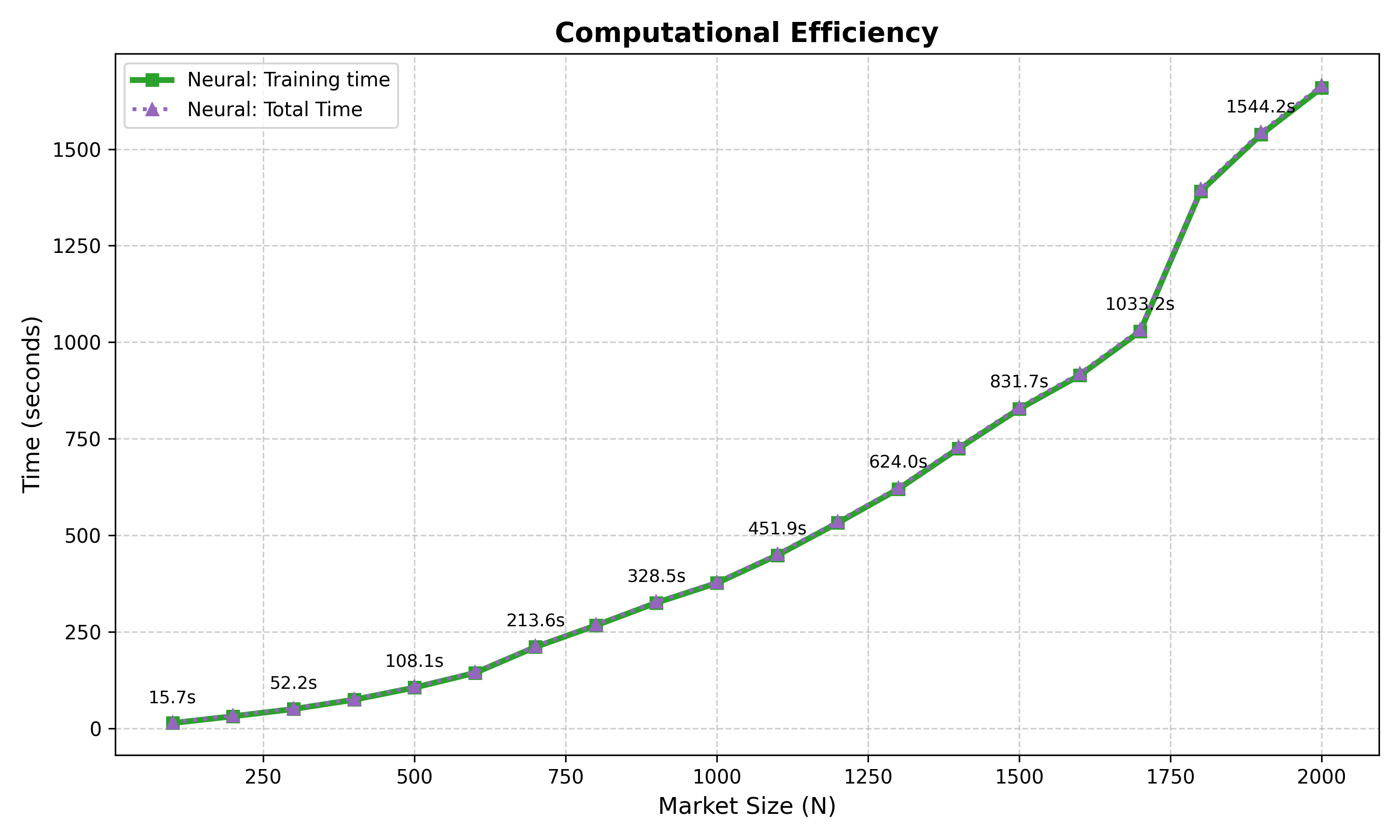}
        \caption{CPU runtime scaling on markets up to 2,000 students (M4 Max).}
    \end{subfigure}

    \vspace{0.5em}

    \begin{subfigure}{0.9\textwidth}
        \centering
        \includegraphics[width=0.8\linewidth]{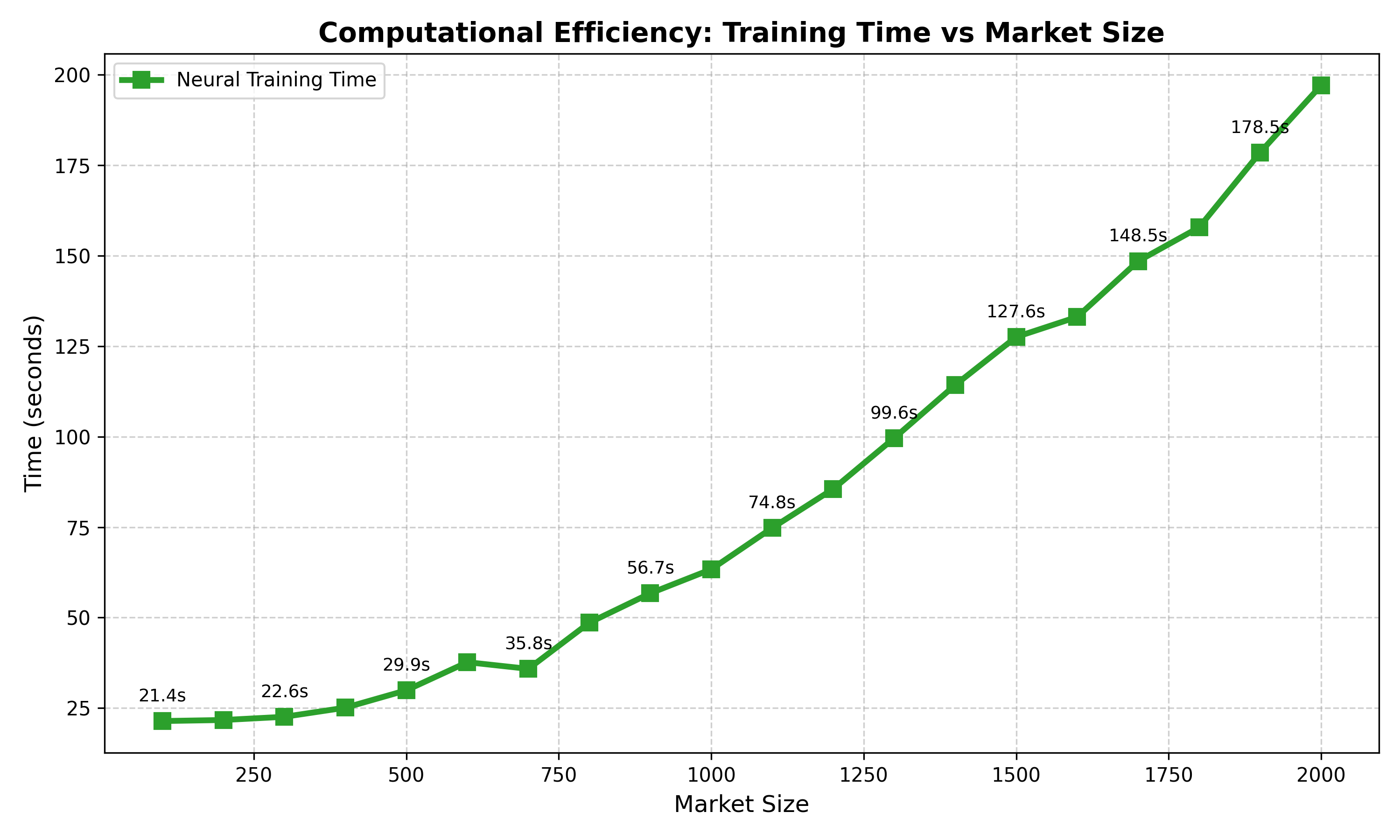}
        \caption{GPU runtime scaling on larger markets up to 2,000 students (RTX 5090).}
    \end{subfigure}

    \caption{Training time scaling with market size.}
    \label{fig:time}
\end{figure}

\begin{figure}[tb]
    \centering
    \includegraphics[width=0.99\textwidth]{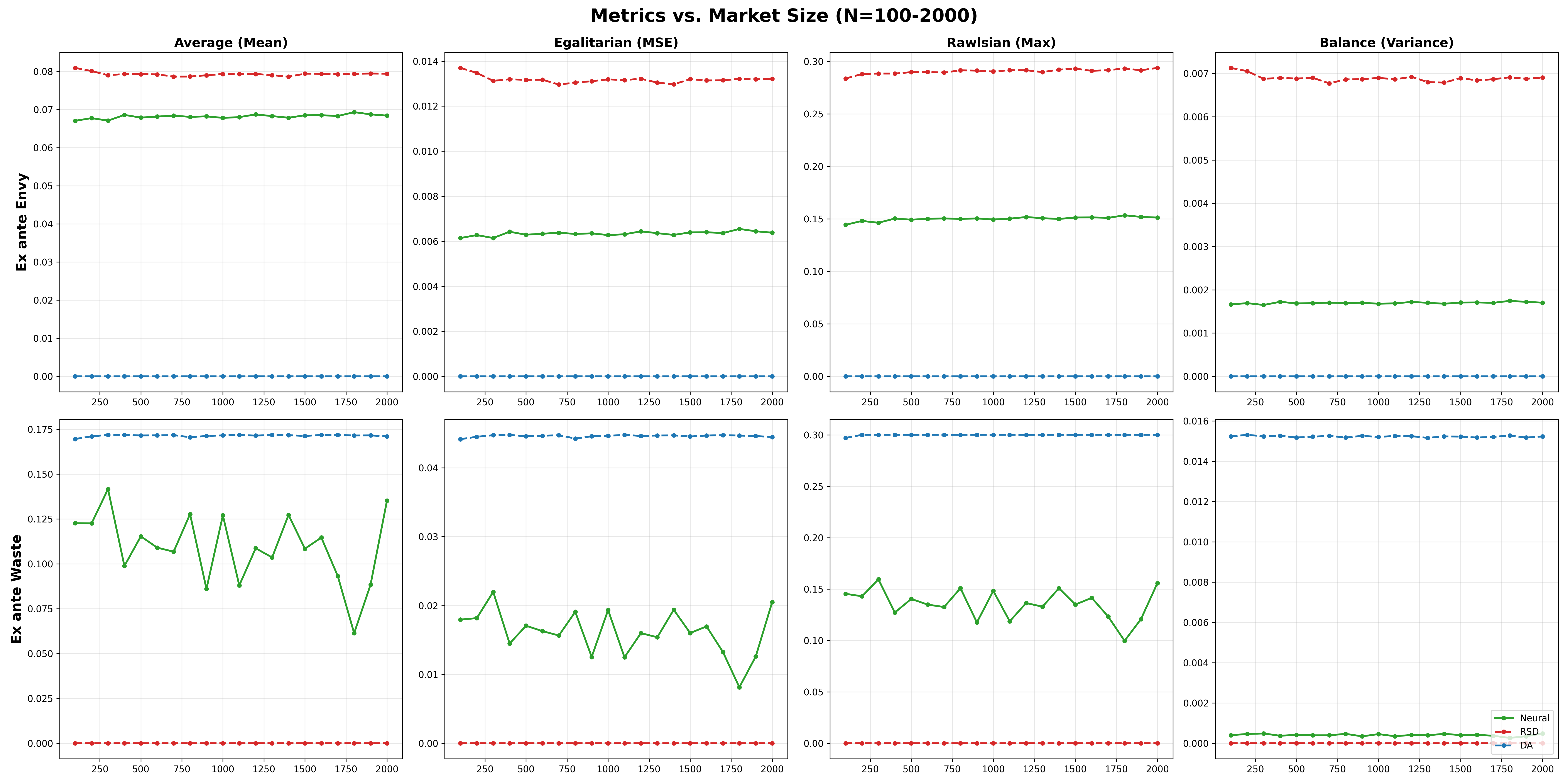}
    \caption{Ex ante fairness and non-wastefulness metrics.}
    \label{fig:metrics}
\end{figure}

\begin{figure}[tb]
    \centering
    \includegraphics[width=0.8\linewidth]{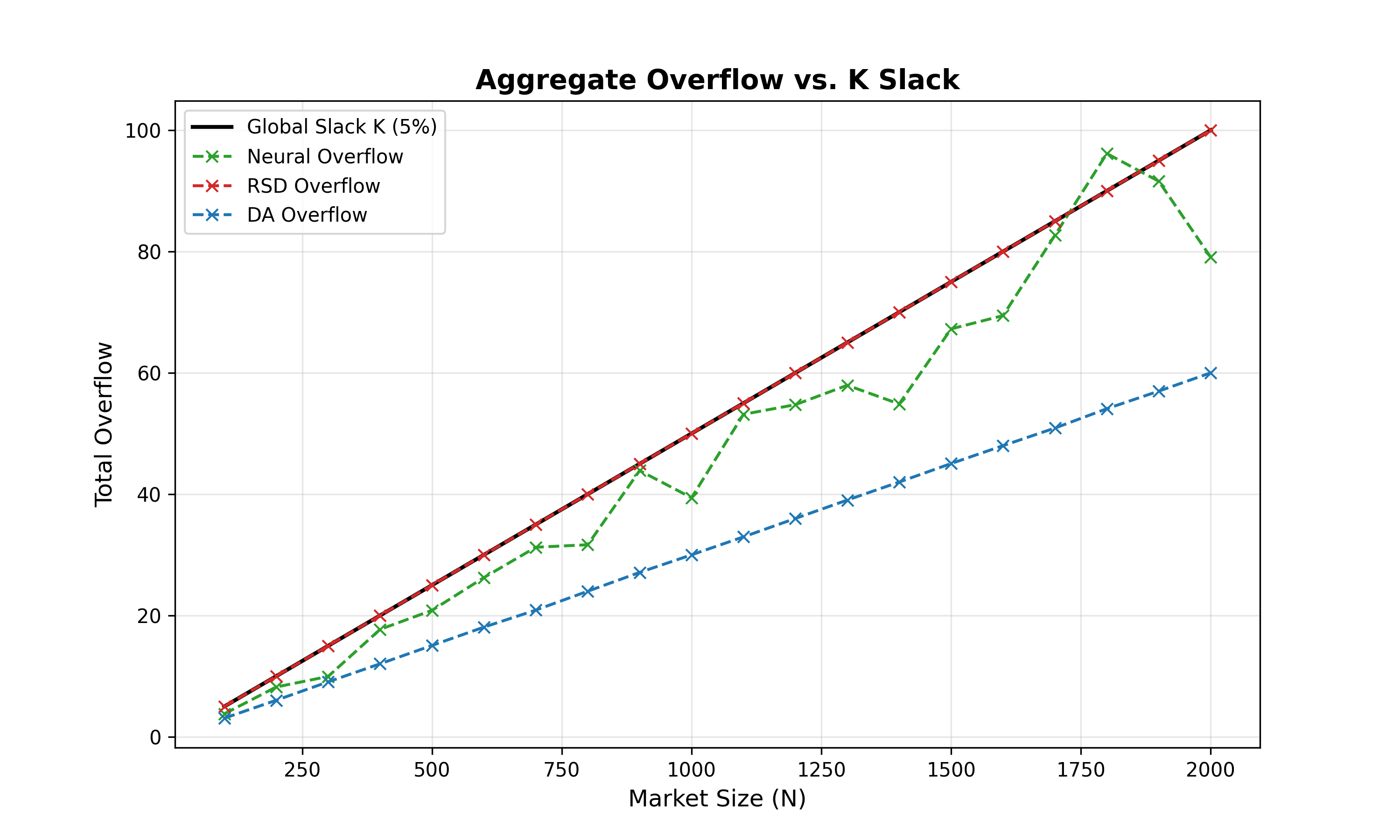}
    \caption{Capacity usage under flexible quotas.}
    \label{fig:overflow}
\end{figure}

\section{Mallows Model and RIM Sampling}
\label{appendix:mallows}
The Mallows model \citep{LuCr11a} is a distance-based exponential family distribution over the space of permutations. Let $\mathcal{S}_m$ denote the set of all permutations over $m$ items. Given a reference ranking $\succ^{\mathrm{ref}} \in \mathcal{S}_m$ and a dispersion parameter $\phi \in (0,1]$, the probability of observing a ranking $\succ \in \mathcal{S}_m$ is defined as:
\begin{equation}
    \mathbb{P}(\succ \mid \phi, \succ^{\mathrm{ref}}) = \frac{1}{Z(\phi)} \phi^{d_{\tau}(\succ, \succ^{\mathrm{ref}})},
\end{equation}
where $d_{\tau}(\cdot,\cdot)$ denotes the Kendall tau distance, which counts the number of pairwise adjacent swaps required to transform one ranking into another. The normalization constant (partition function) $Z(\phi)$ is independent of $\succ^{\mathrm{ref}}$ and admits a convenient closed-form expression:
\begin{equation}
    Z(\phi) = \prod_{i=1}^{m} \sum_{j=0}^{i-1} \phi^j = \prod_{i=1}^{m} \frac{1-\phi^i}{1-\phi}.
\end{equation}
As $\phi \to 0$, the distribution concentrates its mass entirely on the reference ranking $\succ^{\mathrm{ref}}$, whereas $\phi = 1$ recovers the uniform distribution over $\mathcal{S}_m$.

To perform efficient and exact sampling from this distribution, we employ the \emph{Repeated Insertion Model} (RIM). Suppose the reference ranking is fixed as $(c_1, \dots, c_m)$. The sampled ranking is constructed incrementally by inserting each item $c_i$ at position $k \in \{1,\dots,i\}$ with probability:
\begin{equation}
    \mathbb{P}(\text{item } c_i \text{ is at position } k) = \frac{\phi^{i-k}}{\sum_{j=1}^{i} \phi^{i-j}}.
\end{equation}
Starting from $c_1$, each subsequent item $c_i$ is inserted into the existing permutation of $\{c_1, \dots, c_{i-1}\}$. This procedure generates exact samples from the Mallows distribution under the Kendall tau distance in $O(m^2)$ time, which is highly efficient for the market sizes considered in our experiments.

\section{Two Strategy-Proof Baseline Mechanisms}
\label{appendix:baseline}
To benchmark the performance of \texttt{MenuNet}, we consider two classical mechanisms that preserve strategy-proofness while exhibiting complementary trade-offs between fairness and efficiency under global capacity constraints.

\paragraph{Random Serial Dictatorship (RSD).}
We consider a variant of \emph{Random Serial Dictatorship} adapted to the flexible capacity setting. The mechanism first samples a random permutation (master list) over students $\sigma \in \mathcal{S}_N$ with uniform probability. Students are processed sequentially according to $\sigma$. When student $s$ is processed, they are assigned to their most preferred school $c$ that is still \emph{globally feasible}. In this context, a school $c$ remains feasible if the current total enrollment across all schools does not exceed the aggregate capacity $\sum_{j \in C} q_j + K$. 

By construction, this mechanism is strategy-proof, as each student faces a serial decision over a fixed set of available options. Moreover, RSD is \emph{non-wasteful}: if a student can be assigned to a more preferred school without violating the global constraint $K$, the algorithm will prioritize that assignment. However, RSD does not account for school priorities, which typically leads to significant justified envy as earlier students in the ordering can displace those with higher priority. Consequently, while efficient, the mechanism is generally \emph{unfair}.

\paragraph{Deferred Acceptance (DA).}
As a complementary baseline, we employ the \emph{Student-Proposing Deferred Acceptance} algorithm. To incorporate flexibility while preserving strategy-proofness, we must distribute the global slack $K$ across schools \emph{independently} of the students' reported preferences. Specifically, we augment each school $c$'s nominal capacity $q_c$ by a proportional share of the global slack, defined as $q'_c = q_c + \lfloor K / |C| \rfloor$, and distribute the remainder $K \pmod{|C|}$ randomly among the schools. 

The standard DA algorithm is then executed using these adjusted capacities $q'$. This approach ensures that the mechanism remains strategy-proof and eliminates justified envy with respect to the augmented capacities. However, this ex ante allocation of slack is inherently \emph{wasteful}: since the flexibility is distributed uniformly rather than being directed by realized demand, some schools may end up with unused slack while students desiring those seats are rejected elsewhere. This mismatch highlights the inefficiency of rigid slack allocation compared to the adaptive, differentiable approach used in \texttt{MenuNet}.

\section{Implementation Details}
The definitions of envy and waste intensity involve discontinuous indicator functions and are therefore not directly amenable to gradient-based optimization. To address this, we replace them with differentiable surrogates during training.

Specifically, the priority condition $\mathbf{1}[s \succ_c s']$ is approximated by a sigmoid function $\sigma\!\left(\tau_{\mathrm{prio}} \cdot (r_{s,c} - r_{s',c})\right)$, where $r_{s,c}$ denotes the priority score of student $s$ at school $c$ and $\tau_{\mathrm{prio}} > 0$ controls the sharpness of the approximation. Furthermore, we define the \emph{utility gap} as $\Delta_{s,c}(P) = \max\{0, v_{s,c} - \bar{u}_s(P)\}$, representing student $s$'s desire for school $c$ above their current expected utility $\bar{u}_s(P)$.

For waste calculation, we model the availability of seats at school $c$ via a continuous gate:
\begin{equation}
    G_c(P) = \min\left\{1,\, [q_c - \ell_c(P)]_+ + [K - \Omega(P)]_+ \right\},
\end{equation}
where $[x]_+ = \max(0, x)$, $\ell_c(P)$ is the current load, and $\Omega(P)$ is the total system-wide overflow. This gate provides a continuous interpolation between available and saturated regimes. The surrogate waste for student $s$ at school $c$ is then computed as the product of the utility gap and the availability gate: $w_{s,c}(P) = \Delta_{s,c}(P) \cdot G_c(P)$.

These approximations define piecewise differentiable surrogates that preserve the economic intuition of the original definitions while enabling efficient optimization via backpropagation.

\section{Ex post Feasibility and Decomposition Challenges}
Unlike standard capacity constraints that operate independently for each institution, the \emph{global slack} constraint introduces complex interdependencies across the entire market. According to the framework established by \citet{BCK+13a}, such a structure where a global constraint overlaps with local capacity limits, is generally not \emph{universally implementable}. Specifically, a probabilistic assignment matrix $P$ that satisfies the global slack $K$ in expectation may not admit a Birkhoff--von Neumann decomposition into a convex combination of pure matchings that each strictly respect the aggregate bound.

This fundamental theoretical gap presents a significant challenge for traditional combinatorial mechanisms, which typically require a structured hierarchy of constraints to ensure feasible implementation. Our learning-based approach, \texttt{MenuNet}, addresses this limitation by explicitly internalizing these non-linear dependencies through a differentiable $K$-barrier penalty during training. By penalizing the expected violation, \texttt{MenuNet} identifies assignment distributions that are biased toward "near-feasible" regions of the polytope. While the theoretical non-implementability remains, our experiments (see Figure~\ref{fig:overflow}) demonstrate that this approach yields ex post realizations where deviations from the global budget are negligible, effectively navigating a constraint space where exact deterministic implementation is not mathematically guaranteed.

\end{document}